\documentclass[twoside,journal]{IEEEtran}
\usepackage{amsfonts}
\usepackage{amssymb}
\usepackage{amsmath}
\usepackage{mathrsfs}
\usepackage{verbatim}
\usepackage{subfigure}
\usepackage{balance}
\usepackage{booktabs}
\usepackage{fancybox}
\usepackage{bm}
\usepackage{changepage}
\usepackage{extarrows}
\usepackage{algorithm}
\usepackage{algorithmic}
\usepackage{diagbox}
\usepackage{multirow}
\usepackage{array}
\usepackage{graphicx}
\usepackage{epstopdf}
\usepackage{caption}
\usepackage{cite}
\usepackage{color}
\newtheorem{theorem}{Theorem}
\newtheorem{lemma}{Lemma}
\newtheorem{corollary}{Corollary}
\newtheorem{proof}{Proof}
\newtheorem{proposition}{Proposition}

\begin{document}

\title{Wireless Covert Communications \\Aided by Distributed Cooperative Jamming \\over Slow Fading Channels}
\author{Tong-Xing~Zheng,~\IEEEmembership{Member,~IEEE},
Ziteng~Yang,
Chao~Wang,~\IEEEmembership{Member,~IEEE},
Zan~Li,~\IEEEmembership{Senior Member,~IEEE}, ~Jinhong~Yuan,~\IEEEmembership{Fellow,~IEEE},
and~Xiaohong~Guan,~\IEEEmembership{Fellow,~IEEE}
	\thanks{T.-X.~Zheng and Z. Yang are with the School of Information and Communications Engineering, Xi'an Jiaotong University, Xi'an 710049, China. T.-X.~Zheng is also with the National Mobile Communications Research Laboratory, Southeast University, Nanjing 210096, China. T.-X. Zheng and X. Guan are also with the Ministry of Education Key Laboratory for Intelligent Networks and Network Security, Xi'an Jiaotong University, Xi'an 710049, China.
	(e-mail: zhengtx@mail.xjtu.edu.cn, yzy543251162@stu.xjtu.edu.cn, xhguan@mail.xjtu.edu.cn).}
	\thanks{C. Wang and Z. Li are with the State Key Laboratory of Integrated Services Networks, Xidian University, Xi'an 710071, China. C. Wang is also with the National Mobile Communications Research Laboratory, Southeast University, Nanjing 210096, China. (e-mail: drchaowang@126.com, zanli@xidian.edu.cn).}
	\thanks{J. Yuan is with the School of Electrical Engineering and	Telecommunications, University of New South Wales, Sydney, NSW 2052, Australia (e-mail: j.yuan@unsw.edu.au).}
}

\maketitle
\vspace{-0.8 cm}

\begin{abstract}
    In this paper, we study covert communications between {a pair of} legitimate transmitter-receiver against a watchful warden over slow fading channels.
    There coexist multiple friendly helper nodes who are willing to protect the covert communication from being detected by the warden.
    We propose an uncoordinated jammer selection scheme where those helpers whose instantaneous channel gains to the legitimate receiver fall below a pre-established selection threshold will be chosen as jammers radiating jamming signals to defeat the warden.
    By doing so, the detection accuracy of the warden is expected to be severely degraded while the desired covert communication is rarely affected.
    We then jointly design the optimal selection threshold and message transmission rate for maximizing covert throughput under the premise that the detection error of the warden exceeds a certain level. Numerical results are presented to validate our theoretical analyses. It is shown that the multi-jammer assisted covert communication outperforms the conventional single-jammer method in terms of covert throughput, and the maximal covert throughput improves significantly as the total number of helpers increases, which demonstrates the validity and superiority of our proposed scheme.
\end{abstract}


\IEEEpeerreviewmaketitle

\section{Introduction}

With the advent of the 5G information age, the dependence of human being on wireless networks has grown rapidly.
An unprecedentedly large amount of confidential and sensitive information, e.g., ID and credit card details, etc, is being delivered through wireless media.
Security and privacy {have} become ever-increasingly crucial for wireless networks since wireless transmission is particularly prone to security attacks such as eavesdropping and surveillance, etc, due to the broadcast nature of  electromagnetic propagation.
In tradition, the information security issue is addressed at the upper layers of a communication system leveraging the key-based cryptographic encryption \cite{Menezes1996Handbook}, \cite{Talbot2006Complexity}.
However, even the encryption algorithms recognized complicated enough currently are still at risk of being broken by supercomputers and massively parallel computing.
On the other hand, information-theoretic security has appeared as an appealing physical-layer alternative to complement conventional upper-layer security protocols, which aims to {achieve} wireless secrecy by adequately exploiting the intrinsic characteristics of wireless channels, thus circumventing the employment of secrecy keys as well as sophisticated encryption algorithms \cite{Bloch2011Physical, Wang2016Physical, Yang2015Safeguarding,Zou2016Survey}.
It should be pointed out that both upper- and physical-layer security mechanisms have mainly been designed for protecting confidential content from being intercepted by unauthorized third parties, whereas few has been concentrated on shielding the communication itself to escape monitoring from watchful adversaries.

In fact, to guarantee the concealment of message delivery is urgently required for the purpose of a high level of security for various wireless applications, e.g., secret military operations, authoritarian government monitoring, and internal interaction of the Internet of Things (IoTs) \cite{Mukherjee2015Physical}, etc.
The exposure of communication process will cause suspicion at the monitor and hence attract further signal demodulation and information decipher.
In order to overcome the above problem and to provide a more direct and powerful safeguard on security and privacy for wireless communications systems, covert communications has emerged as a new paradigm for wireless security and {it} has attracted significant research interest recently \cite{Yan2019Low}.

\subsection{Previous Works}
Covert communications aims to hide the very existence of communication behaviors from a vigilant warden, i.e., guaranteeing a {very} low probability of being detected by the warden, while attaining a certain level of decoding performance for the target receiver.
Although spread-spectrum techniques have been widely used to obscure military communications against adversaries in the early 20 century \cite{Simon1994Simon}, many fundamental issues have not been well tackled, e.g., it is still challenging to answer for spread-spectrum techniques that \emph{how much information can be successfully delivered subject to a certain level of covertness requirement?}
It is until recently that the fundamental limits of covert communications have been characterized from an information-theoretic viewpoint. In \cite{Bash2013Limits}, a square root law was first established for covert communications in additive Gaussian white noise (AWGN) channels, which states that at most {$\mathcal{O}(\sqrt{n})$} information bits can be delivered reliably and covertly simultaneously to the intended receiver in $n$ channel uses under the surveillance of a warden.
This pioneering work has laid an information-theoretic foundation for the research of covert communications and has opened up two research directions: 1) explore fundamental limits of covert communications based on an information-theoretic framework and devise practical code structures to realize covert communications; 2) excavate the feasibility of covert communications for realistic circumstances and exploit advanced techniques to improve covert communications.

Along the first research direction mentioned above, researchers have generalized the work of \cite{Bash2013Limits} into various wireless channel models, such as binary symmetrical channels (BSC) \cite{Che2013Reliable}, discrete memoryless channels (DMC) \cite{Bloch2016Covert,Wang2016Fundamental}, multiple access channels \cite{Arumugam2016Keyless}, parallel AWGN channels \cite{Letzepis2016Optimal}, and multi-input multi-output (MIMO) AWGN channels \cite{Abdelaziz2017Fundamental}, etc.
These works have been focused on capturing the scaling law of the amount of covertly delivered information on $\sqrt{n}$ and examining the feasibility of the square root law proposed in \cite{Bash2013Limits}.
It is not difficult to find that, according to the square root law, the achievable covert rate would approach zero as $n\rightarrow\infty$, which is definitely undesired for practical applications.
In fact, such a square root law is overly pessimistic as the intrinsic uncertainty of wireless channels and noise has not been incorporated therein, which then has stimulated the second research direction as will be detailed below.

Recently, extensive efforts have been devoted to exploring the opportunity and condition for attaining a positive covert rate.
Specifically, the authors in \cite{Lee2015Achieving} and \cite{Goeckel2016Covert} proved that a positive covert rate can be achieved when the background noise power is unknown to the warden.
The authors in \cite{He2017On} analyzed the maximal achievable covert rate under both bounded and unbounded noise uncertainty models.
The authors in \cite{Shahzad2017Covert,Yan2017Covert,Tang2018Covert,Yan2018Delay} evaluated the influence of imperfect channel state information (CSI) and finite blocklength on covert communications, respectively.
In particular, the authors in \cite{Yan2018Delay} proved that increasing the blocklength is always rewarding for covert throughput for a delay-sensitive system, and adopting a random transmit power can provide additional performance gain.
The authors in \cite{He2018Covert} exploited ambient interference signals to improve the rate of covert communications.
The authors in \cite{Zheng2019Multi} further explored the benefits of both centralized and distributed multi-antenna transmitters on covert communications against randomly located wardens.
The authors in a recent work \cite{Shahzad2019Covert} considered a multi-antenna warden scenario, showing that even a slight increase in the antenna number at the warden side can result in a dramatic deterioration in covert throughput.

Node cooperation has also been leveraged to enable and enhance covert communications recently.
For instance, the authors in \cite{Wang2019Covert} investigated the dual-hop covert communication under fading channels.
It was shown that with the aid of a relay and the exploitation of channel uncertainty, $\mathcal{O}(n)$ information bits can be conveyed reliably and covertly in $n$ channel uses, which broke through the square root law proposed in \cite{Bash2013Limits}.
The authors in \cite{Sheikholeslami2018Multi} examined multi-hop routing in covert wireless networks under AWGN channels. The authors in \cite{Wang2020Secrecy} further studied secrecy and covert communications via a multi-hop networks under the surveillance of an unmanned aerial vehicle (UAV).
It is worth mentioning that, instead of assisting the covert communication, relay nodes might also privately send messages to the destination node by greedily utilizing the resources provided by the source node.
In this case, the source node becomes a warden, whose mission is to monitor in real-time the covert communication from the untrusted relay. The interested readers are referred to \cite{Hu2018Covert,Hu2019Covert} for more details in this direction.

\subsection{Motivation and Our Contributions}
The aforementioned endeavors on covert communications have mainly focused on taking advantage of already-existed favorable environmental uncertainty, such as noise, channel uncertainty, and auxiliary relays, etc.
Different from these, recent studies have revealed that intentionally {introducing} interference, e.g., via a friendly jamming node using random transmit power \cite{Sobers2017Covert} or via a full-duplex receiver \cite{Hu2018Covert_C,Shahzad2018Achieving,Shu2019Delay} with either fixed or random jamming power, could also be beneficial for covert communications, since additional uncertainty and thus more confusion can be created to the warden when {the warden} is analyzing the statistics of perceived energy.
Inspired by this idea, in this paper we, for the first time, investigate covert communications assisted by multiple friendly jammers over slow fading channels. We aim to excavate more potential of cooperative jamming in terms of improving covert communications compared with that of a single helper.
We propose to intelligently select jammers based on their individual channel status in order to balance covertness and throughput performance.
To the best of our knowledge, this idea has not yet been reported by the literature on covert communications.
The main contributions of our work are summarized as follows.

\begin{itemize}
\item {\color{black}We put forward an uncoordinated opportunistic jammer selection scheme to facilitate the covert communication against a warden over slow fading channels.
	To be specific, in order to degrade the detection performance of the warden while not leaving severe interference to the desired receiver, each cooperative node is prudently activated as jammer only when its instantaneous channel gain with respect to the destination receiver becomes smaller than a certain threshold.
	By doing so, even the jammers use fixed transmit power, the warden is still uncertain of the aggregate interference power due to the random number of actually activated jammers. In addition, the covertness and reliability performance can be well balanced by properly adjusting the selection threshold.}

\item Under the proposed jammer selection scheme, we first examine the worst-case scenario for covert communications where the warden is always capable of employing the optimal detection parameter to minimize detection errors.
In this scenario, we derive tractable expressions for both the detection error probability at the warden and the transmission outage probability at the desired receiver.
Subject to a covertness constraint in which the average  detection error probability is restricted to exceed a certain value, we then jointly design the optimal selection threshold and transmission rate to maximize the covert throughput.

\item We present adequate numerical results to validate our theoretical findings.
By comparing with the single-jammer aided covert communications, we demonstrate the superiority of our proposed multi-jammer scheme in terms of covert throughput enhancement.
We also provide various useful insights into the multi-jammer aided covert communication.
In particular, we show that, as the number of cooperative nodes increases, the covert throughput first significantly improves but then reaches saturation.
It is interesting to find that, in order to meet a certain level of covertness requirement, there exists a minimum number of cooperative nodes that should be deployed, and the minimum number increases as the jamming power becomes smaller.

\end{itemize}

\subsection{Organization and Notations}
The remainder of this paper is organized as follows.
Section II details the system model and relevant hypotheses.
Section III analyzes the detection error probability at the warden considering a worst-case scenario for covert communications.
In Section IV, we first derive analytical expressions for the average detection error probability and the transmission outage probability, based on which we then optimize the covert throughput subject to a covertness constraint.
Numerical results are presented in Section V, followed by some insights into parameter designs.
Section VI concludes the paper.

\emph{Notations}:
Bold lowercase letters denote column vectors.
$|\cdot|$, $(\cdot)^{\dag}$, $\mathbb{P}[\cdot]$, and $\mathbb{E}_v[\cdot]$ denote the absolute value, conjugate, probability, and mathematical expectation taken over a random variable $v$, respectively.
$\mathcal{CN}(u,\sigma^2)$ denotes the circularly symmetric complex Gaussian distribution with mean $u$ and variance $\sigma^2$. $\binom{N}{m}\triangleq \frac{N!}{m!(N-m)!}$ denotes the combinatorial number. $\Gamma(a)\triangleq\int_{0}^{\infty}e^{-t}t^{a-1}dt$, $\Gamma\left(a,x\right)\triangleq\int_{x}^{\infty}{e^{-t}t^{a-1}}dt$, and $\gamma\left(a,x\right)\triangleq\int_{0}^{x}{e^{-t}t^{a-1}}dt$ denote the gamma function, the upper incomplete gamma function, and the lower incomplete gamma function, respectively.

\section{System Model}
\subsection{Communication Scenario and Assumptions}
\begin{figure}[!t]
\centering
\includegraphics[width =3.5in]{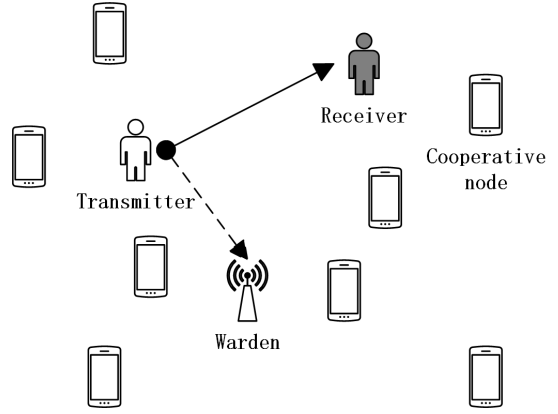}
\caption{Illustration of  covert communications in a multi-node cooperative network, where multiple Cooperative nodes help to achieve covert communications between Transmitter and Receiver in the presence of Warden.}
\label{model tu}
\vspace{-0.0 cm}
\end{figure}

We consider a cooperative covert communication network, where transmitter T aims to covertly deliver messages to receiver R under the surveillance of warden W, as depicted in Fig. \ref{model tu}.
There coexist a set of $N$ friendly nodes ${\rm J}_k$, for $k=1,\cdots,N$, who collaborate to shield the ongoing covert communication from T to R to make it as invisible as possible at W.
Specifically, a portion of the cooperative nodes will be judiciously selected acting as jammers for sending interference signals to confuse W, as will be described in the next subsection.
We assume that each node in the network is equipped with a single antenna.

To model the wireless channel between two nodes, we adopt a combined channel model incorporating a standard distance-based path loss governed by the exponent $\alpha$ along with Rayleigh fading.
We consider channel reciprocity, and the channel gains between T and R, T and W, ${\rm J}_k$ and R, and ${\rm J}_k$ and W can be expressed as $h_{t,r}d_{t,r}^{-{\alpha}/{2}}$, $h_{t,w}d_{t,w}^{-{\alpha}/{2}}$, $h_{j_k,r}d_{j_k,r}^{-{\alpha}/{2}}$ , and $h_{j_k,w}d_{j_k,w}^{-{\alpha}/{2}}$, respectively, where $h$ denotes the independent and identically distributed (i.i.d.) fading coefficient obeying $\mathcal{CN}(0,1)$, and $d$ denotes the corresponding path distance.
The network topology, including the locations of nodes and the distances between each other, is considered to be overt to all the nodes in the network.\footnote{\color{black}This assumption can be tenable for some real-life circumstances. Typical examples include that W is a government regulatory agency whose location is public to all, or communication entities are deployed on the top of conspicuous buildings such that they are visible and the distance can be easily measured. Moreover, it is a worst case that assuming W knows the location information of all nodes, which is preferred for a robust design.}

{\color{black}We consider a two-phase channel estimation scheme.
	Specifically, in the first phase, R broadcasts a public pilot sequence to enable all the $N$ friendly nodes ${\rm J}_k$ to know the instantaneous CSI of their channels to R, i.e., $h_{j_k,r}$, which is required for the jammer selection scheme as will be described in the next subsection;\footnote{We assume that Willie only focuses on monitoring the covert communication from T to R while R can be overt to Willie, and hence there is no covertness requirement for the pilot delivery of R. }
	in the second phase, T covertly sends a secret pilot sequence shared between itself and R to facilitate R to estimate their channel $h_{t,r}$ for decoding signals, and meanwhile the selected jammers send different pseudo-random pilot sequences to conceal T's pilot sequence against Willie.\footnote{As noted in \cite{Bash2013Limits}, T and R possess a common secret randomness resource in order to enable covert communications. In our channel estimation scheme, such a resource is a secret pseudo-random pilot sequence shared between T and R prior to communication which is unknown to Willie. As long as the sequence length is sufficiently large, it will become nearly orthogonal to those pilot sequences sent by jammers.
In this way, R can easily match T's pilot sequence and extract the channel $h_{t,r}$, while Willie cannot detect the pilot delivery of T  due to the cover of jamming pilot sequences.}}
We also examine two situations where W can and cannot acquire the instantaneous CSI of the channel between itself and T, i.e., $h_{t,w}$, as will be detailed in In Sec. III-B.\footnote{By assuming the availability of the instantaneous value of $h_{t,w}$ at W, we expect to explore the worst performance of covert communications.}

\subsection{Jammer Selection Scheme}
In order to conceal the covert communication against W, an uncoordinated jammer selection scheme is employed, where those cooperative nodes with their individual power gains related to R being smaller than a certain threshold $\tau$ will be activated to transmit jamming signals independently; whilst the others remain silent.\footnote{{\color{black}Our scheme requires only a low-level collaboration. This is fundamentally different from those higher-level collaboration schemes such as distributed zero-forcing jamming, which require to collect global CSI and optimize the collaborative beamformer weights coordinately, thus resulting in  a high overhead and implementation complexity due to signal sharing and synchronization.}}
Mathematically, the threshold-based jammer selection criterion can be described by the following indicator function
\begin{align}\label{characteristic I_k}
       \mathcal{I}_{j_k}=\begin{cases}
       1,&{\left|h_{j_k,r}\right|^2d_{j_k,r}^{-\alpha}}/{\sigma_r^2}<\tau,\\
       0,&{\left|h_{j_k,r}\right|^2d_{j_k,r}^{-\alpha}}/{\sigma_r^2}\ge\tau,
       \end{cases}
\end{align}
where $\mathcal{I}_{j_k}=1$ indicates that $\rm J_k$ is selected as jammer to send jamming signals, and $\mathcal{I}_{j_k}=0$ corresponds to a silent node.
Note that the proposed jammer selection criterion takes both small-scale channel fading and large-scale path loss into consideration.
In addition, the covert communication without jammers assistance or with all jammers being selected are two special cases of our scheme with $\tau=0$ and $\tau\rightarrow\infty$, respectively.

We should emphasize that, in order to confuse W, all the cooperative nodes will continuously choose to send jamming signals or keep silent according to the jammer selection criterion given in \eqref{characteristic I_k}, no matter whether T is transmitting or not.
The reason behind such a threshold-based jammer selection criterion is that there is a trade-off between the covertness and reliability for covert communications.
It is obvious that, setting a large selection threshold $\tau$ helps to activate a large number of jammers, which is beneficial for achieving covertness since in this way W will be severely confused about whether his perceived increased energy comes from T or the jammers.
In this sense, for the purpose of covertness, it is preferred to choose the value of $\tau$ as large as possible to dramatically degrade the  detection performance of W.
As a coin has two sides, an overly high selection threshold $\tau$ will definitely cause severe interference to the desired receiver, thus impairing transmission reliability.
As will be discussed later for parameter optimization, only by properly designing the selection threshold, can we strike a good balance between transmission covertness and reliability so as to achieve a high level of covert throughput.

Considering quasi-static Rayleigh fading, the cumulative distribution function (CDF) of $| h_{j_k, r} | ^ 2$ in \eqref{characteristic I_k} is denoted as $F_ {| h_ {j_k, r} | ^ 2} (x) = 1-e ^ {-x}$, and accordingly the probability of selecting ${\rm J}_k$ as jammer can be given by
\begin{equation}\label{p}
       p_{j_k} = 1 - e^{-d_{j_k,r}^{\alpha}{\sigma_r^2}\tau}.
\end{equation}

If T transmits messages in a certain time slot $i$, the signal received at R can be expressed as follows,
\begin{equation}\label{yr}
   \bm{y}_r\left[i\right] = \sqrt{P_t}\frac{h_{t,r}}{d_{t,r}^{\alpha/2}}\bm{x}_t\left[i\right]
   + \sum_{k=1}^{N}{\mathcal{I}_{j_k}\sqrt{P_j}\frac{h_{j_k,r}}{d_{j_k,r}^{\alpha/2}}\bm{v}_{j_k}\left[i\right]} + \bm{n}_r\left[i\right],
\end{equation}
where $\bm{x}_t[i]$ is the {\color{black}Gaussian} signal transmitted by T satisfying $ \mathbb{E}[\bm{x}_t[i]\bm{x}_t^{\dag}[i]] = 1$, $i = 1,2, \cdots, n$, with $n$ being the number of channel uses for delivering a whole codeword, and $P_t$ denotes the transmit power.
$\bm{v}_ {j_k}[i]$ is the {\color{black}Gaussian} jamming signal radiated by ${\rm J}_k$ in the current time slot satisfying $\mathbb{E}\left[\bm{v}_{j_k}[i] \bm{v}_{j_k}^{\dag} [i]\right] = 1$, and $P_j$ denotes the jamming power. $\bm{n}_r[i]$ is the AWGN at R with variance $\sigma_r ^ 2$, i.e., $\bm{n}_r [i ] \sim \mathcal {CN} (0, \sigma_r^2 )$.

It is worth noting that, compared with conventional jammer-assisted covert communications where the jammer uses random transmit power to confuse W, in this paper we simply consider the same fixed jamming power for all the helper nodes.
Nevertheless, we can still create uncertainty for W in terms of his received power with the aid of the proposed opportunistic jammer selection scheme.
Since W is unaware of the instantaneous channel gains between ${\rm J}_k$ and R, even the value of {\color{black}selection} threshold $\tau$ is known by W, W still cannot determine the exact number of jammers, let alone the aggregate interference power.
As a consequence, W is not able to judge whether the increased received power is due to the start of covert communications or just the activation of new jammers.

\subsection{Detection Strategy at Warden}
In this subsection, we will detail the detection strategy for W. Note that W attempts to detect whether the covert communication between T and R takes place or not in a certain time slot.
Thus, W faces a binary hypothesis testing problem, where the null hypothesis $\mathcal {H} _0$ states that T is not transmitting, while the alternative hypothesis $\mathcal{H} _1$ states that T is communicating to R. Under the above two hypotheses, the signal received at W can be expressed as below
\begin{align}\label{y_e}
       \bm{y}_w[i]=
       \begin{cases}
       \sum\limits_{k=1}^{N}{\mathcal{I}_{j_k}\sqrt{P_j}\frac{h_{j_k,w}}{d_{j_k,w}^{\alpha/2}}\bm{v}_{j_k}[i]}+\bm{n}_w[i],&\mathcal {H} _0,\\
       \sqrt{P_t}\frac{h_{t,w}}{d_{t,w}^{\alpha/2}}\bm{x}_t[i]+
       \sum\limits_{k=1}^{N}{\mathcal{I}_{j_k}\sqrt{P_j}\frac{h_{j_k,w}}{d_{j_k,w}^{\alpha/2}}\bm{v}_{j_k}[i]}+\bm{n}_w[i],&\mathcal {H} _1,
       \end{cases}
\end{align}
where $\bm{n}_w [i]$ is the AWGN at W with variance $\sigma_w ^ 2$, i.e., $\bm{n}_w [i]\sim\mathcal {CN} (0, \sigma_w ^2 )$.

{\color{black}We follow a commonly used assumption in literature on covert communications that W adopts a radiometer to detect the covert communication from T to R \cite{He2017On,He2018Covert,Zheng2019Multi,Shahzad2019Covert,Shahzad2017Covert}.
	This is not only because the radiometer is easy to implement and is of low complexity in practice, but also because of its sufficiency in detecting covert communications for the scenario under investigation as justified in Appendix \ref{proof_sufficiency_radiometer}.
	}
Specifically, W employs the average received power $T_w$ with $n$ channel uses as the test statistic, which is given below
\begin{equation}\label{average T_e}
   T_w=\frac{1}{n}\sum\limits_{i=1}^{n}\left|\bm{y}_w[i]\right|^2.
\end{equation}
The ultimate goal of W is to determine whether $\bm{y}_w$ is under $\mathcal {H} _0$ or under $\mathcal {H} _1$ by analyzing $T_w$.
{According to the mechanism of radiometer, the decision criterion is characterized as follows}
\begin{equation}\label{T_e_mu}
   T_w\mathop{\gtrless}\limits_{\mathcal{D}_0}^{\mathcal{D}_1}\mu,
\end{equation}
where $\mu$ is a predetermined detection threshold, and $\mathcal {D} _0$ and $\mathcal {D} _1$ indicate that W makes decisions in favor of $\mathcal{H} _0$ and $\mathcal{H} _1$, respectively.
We also adopt a common assumption that W is capable of using an infinite number of signal observations for ease of detection \cite{He2017On,He2018Covert,Zheng2019Multi,Shahzad2019Covert,Shahzad2017Covert},  i.e., $n \rightarrow \infty$, which leads to
\begin{align}\label{T_e}
       T_w=\begin{cases}
       \sigma_j^2+\sigma_w^2,&\mathcal{H} _0,\\
       \sigma_c^2+\sigma_j^2
       +\sigma_w^2,&\mathcal{H} _1.
       \end{cases}
\end{align}
where $\sigma_c^2\!\!\triangleq\!\! P_t{\left|h_{t,w}\right|^2}{d_{t,w}^{-\alpha}}$ and $\sigma_j^2\!\!\triangleq \!\! \sum_{k=1}^{N}{{\mathcal{I}_{j_k}}P_j{\left|h_{j_k,w}\right|^2}{d_{j_k,w}^{-\alpha}}}$ denote the power of the covert signal and the aggregate jamming signals, respectively.

In this paper, we employ the detection error rate to measure the detection performance of W, which consists of the false alarm rate and the miss detection rate. Specifically, the false alarm rate is defined as the probability that W makes a decision in favor of $\mathcal {D} _1$ while $\mathcal {H} _0$ is true, which is denoted by $\mathcal{P}_{FA} \triangleq \mathcal {P} (\mathcal {D} _1 | \mathcal {H} _0)$; the miss detection rate is defined as the probability that W makes a decision in favor of $\mathcal {D} _0$ while $\mathcal {H} _1$ is true, and we denote it as $\mathcal{P}_{MD} \triangleq \mathcal {P } (\mathcal {D}_0 | \mathcal{H}_1)$. We assume that the priori probabilities of hypotheses $\mathcal {H} _0$ and $\mathcal {H} _1$ are equal.
Under this assumption, the detection error probability of W is defined as
\begin{equation}\label{detection error rate}
   \xi \triangleq \mathcal{P}_{FA}+\mathcal{P}_{MD}.
\end{equation}

In practice, it is impossible, or at least difficult, to know the exact value of the detection threshold $\mu$ set by W.
For the purpose of a robust design, we examine a worst-case scenario for covert communications where W can always use the optimal detection threshold to achieve a minimum detection error probability, denoted as ${\xi}^*$.
However, T cannot acquire the instantaneous CSI of the channel between itself and W $|h_{t,w}|^2$, it is impossible to calculate the exact detection error probability.
Instead, we employ the average detection error probability ${\bar{\xi}}$ to quantify the covertness performance of the system, which is averaged over $|h_{t,w}|^2$.
\subsection{Optimization Problem}
In this paper, we focus on the core metric termed covert throughput, which is defined as  \cite{Zheng2019Multi}
\begin{equation}\label{R c def}
\Omega\triangleq R\left(1-\delta\right),
\end{equation}
where $R$ denotes the transmission rate of messages adopted at T, and $\delta$ denotes the transmission outage probability {of the covert communication from T to R}, which will be detailed in Sec. IV-B, quantifying the likelihood that the reliability of the covert communication with rate $R$ cannot be guaranteed over fading channels.

The goal of this work is to achieve the maximal covet throughput under the premise that a certain level of covertness can be promised.
Specifically, we aim at maximizing the covert throughput $\Omega$ subject to a covertness constraint ${\bar{\xi}}^* \geq 1- \varepsilon$ by jointly designing the optimal transmission rate $R$ and the jammer selection threshold $\tau$.
Here, ${\bar{\xi}}^*$ denotes the minimum average detection error probability {at W} for the worst-case scenario of covert communications, and $\varepsilon \in [0,1]$ is a prescribed threshold corresponding to the maximal acceptable probability of being detected by W.
Mathematically, the problem of maximizing the covert throughput can be formulated as follows
\begin{equation}\label{R c max}
\max\limits_{ R,\tau}\Omega=R\left(1-\delta\right),\quad
\rm{s.t.}\quad {\bar{\xi}}^*\geq\rm{1}-\varepsilon.
\end{equation}

\section{Analysis of Covertness Performance}
In this section, we first derive the expression for the detection error probability as defined in \eqref{detection error rate} by calculating the false alarm rate and miss detection rate, respectively, considering a given value of $h_{t,w}$.
Afterwards, we consider a worst-case covert communication scenario where W can always adjust the detection threshold $\mu$ to fading channels and network geometry for improving its detection accuracy, and then design the optimal detection threshold to minimize the average detection error probability, which is averaged over $h_{t,w}$.
\subsection{Detection Error Probability}
Revisiting the detection strategy at W described in Sec. II-C, once given a detection threshold $\mu$, we {need to} derive the false alarm rate $\mathcal{P}_{FA}$ and miss detection rate $\mathcal{P}_{MD}$, respectively.
One of the major obstacles to the calculation of $\mathcal{P}_{FA}$ and $\mathcal{P}_{MD}$ lies in the term, i.e., the aggregate jamming power $\sigma_j^2$ in \eqref{T_e}.
First, the indicator $\mathcal{I}_{j_k}$ is closely related to the proposed opportunistic jammer selection scheme \eqref{characteristic I_k} such that it is a random variable strongly coupled with another random herein, i.e., the channel gain $|h_{j_k,w}|^2$.
Moreover, the distribution of the summation $\sigma_j^2$ is intractable to derive, since the terms therein are independent but non-identically distributed (i.n.i.d.) due to different distances $d_{j_k,w}$.

In order to overcome the above issues, we note that the combinations of selected jammers might vary in different time slots, and introduce the notation $\phi_m^s$ to represent the $s$-th subset among  $\binom{N}{m}$ subsets of different combinations of selected jammers with cardinality $m$ ($m=0,1,2,\ldots,N$), and the corresponding probability is denoted as $\mathcal{P}_{\phi_m^s}$.
For completeness, we denote $\phi_0^1$ as null set indicating no jammer is selected.
In the following lemma, we calculate the probability $\mathcal{P}_{\phi_m^s}$.
\begin{lemma}\label{jammer_set_pro}
	\textit{The probability that the set of the selected jammers is $\phi_m^s$ {is} given by}
	\begin{align}
	 \mathcal{P}_{\phi_m^s}
	 =\begin{cases} \prod\limits_{k=1}^N\left(1-p_{j_k}\right),&m=0\\ \prod\limits_{j_q\in\phi_m^s}p_{j_q}\prod\limits_{j_l\notin\phi_m^s}\left(1-p_{j_l}\right),&m=1,\cdots,N-1\\
	 \prod\limits_{k=1}^Np_{j_k},&m=N\\
	 \end{cases}
	\end{align}
	where $p_{j_k}$ denotes the probability that ${\rm J}_k$ is chosen as jammer, which has been derived in \eqref{p}.
\end{lemma}
\begin{proof}
	The result follows directly by noting the fact that all the jammers are activated independently.
\end{proof}

By introducing $\phi_m^s$, the aggregate jamming power $\sigma_j^2$ perceived at W for a certain time slot with $m$ jammers being activated can be rewritten as $\sum_{j_q\in\phi_m^s}P_j{\left|h_{j_q,w}\right|^2d_{j_q,w}^{-\alpha}}$.
Further, we define $X_{m,s}\triangleq\sum_{j_q\in\phi_m^s}{\left|h_{j_q,w}\right|^2d_{j_q,w}^{-\alpha}}$, and leverage a common gamma approximation method \cite{heath2013modeling} to derive its distribution by treating it as a gamma random variable.
\begin{lemma}\label{gamma_approx}
\textit{	By regarding $X_{m,s}\triangleq\sum_{j_q\in\phi_m^s}{\left|h_{j_q,w}\right|^2d_{j_q,w}^{-\alpha}}$ as a gamma random variable, its probability density function (PDF) can be {approximated} as
	\begin{equation}\label{pdf Xmi}
	f_{X_{m,s}}\left(x_{m,s};v_{m,s},\omega_{m,s}\right)\approx
	\frac{x_{m,s}^{v_{m,s}-1}e^{-\frac{x_{m,s}}{\omega_{m,s}}}}{\omega_{m,s}^{v_{m,s}}\Gamma\left(v_{m,s}\right)},
	\end{equation}
	where $v_{m,s}$ and $\omega_{m,s}$ are respectively defined as below
	\begin{align}
	v_{m,s}\triangleq \frac{\left(\sum_{j_q\in\phi_m^s} d_{j_q,w}^{-\alpha}\right)^2}{\sum_{j_q\in\phi_m^s} d_{j_q,w}^{-2\alpha}},\quad
	\omega_{m,s}\triangleq \frac{\sum_{j_q\in\phi_m^s} d_{j_q,w}^{-2\alpha}}{\sum_{j_q\in\phi_m^s} d_{j_q,w}^{-\alpha}}.
	\end{align}}
\end{lemma}
\begin{proof}
	The values of $v_{m,s}$ and $\omega_{m,s}$ can be obtained by matching the first and second moments of $X_{m,s}$, and a detailed proof is relegated to Appendix \ref{proof_lemma_2}.
\end{proof}

{The accuracy of the approximation \eqref{pdf Xmi} will be confirmed later by simulation results, and hence in the following analysis we will replace the exact PDF of the random variable $X_{m,s}$ with \eqref{pdf Xmi} for mathematical tractability.}
Resorting to Lemmas \ref{jammer_set_pro} and \ref{gamma_approx}, we can calculate the false alarm rate and miss detection rate in the following theorem.
\begin{theorem}\label{theorem_false and miss}
    \textit{The false alarm rate $\mathcal{P}_{FA}$ and the miss detection rate $ \mathcal{P}_{MD}$ at W are respectively given by
    \begin{align}\label{false alarm}
       \mathcal{P}_{FA} &=\begin{cases}
       1,&\mu \leq \sigma_w^2,\\
       \sum\limits_{m=1}^{N}\sum\limits_{s=1}^{\binom{N}{m}}\mathcal{P}_{\phi_m^s}
       \frac{\Gamma\left(v_{m,s},\frac{\mu-\sigma_w^2}{\omega_{m,s}P_j}\right)}{\Gamma\left(v_{m,s}\right)},&\mu>\sigma_w^2,\\
       \end{cases}\\
       \label{miss detection}
       \mathcal{P}_{MD} &=\begin{cases}
       0,&\mu \leq \rho_1,\\
       \mathcal{P}_{\phi_0^1}+\sum\limits_{m=1}^{N}\sum\limits_{s=1}^{\binom{N}{m}}\mathcal{P}_{\phi_m^s}\frac{\gamma\left(v_{m,s},\frac{\mu-\rho_1}{{\omega_{m,s}P}_j}\right)}{\Gamma\left(v_{m,s}\right)},&\mu>\rho_1,
       \end{cases}
    \end{align}
where $\mathcal{P}_{\phi_m^s}$ has been  provided by Lemma \ref{jammer_set_pro} and $\rho_1\triangleq \sigma_c^2+\sigma_w^2$}.
\end{theorem}
\begin{proof}
Recalling the detection strategy described in \eqref{T_e_mu}, the false alarm rate can be calculated from \eqref{T_e} as
\begin{align}\label{alpha proof}
&\mathcal{P}_{FA}=\mathbb{P}\left[\sigma_j^2+\sigma_w^2>\mu\right]\notag\\
  &\stackrel{\mathrm{(a)}}
  =\begin{cases}
  1,&\mu \leq \sigma_w^2,\\
  \sum\limits_{m=0}^{N}\sum\limits_{s=1}^{\binom{N}{m}}\mathcal{P}_{\phi_m^s}\mathbb{P}\left[\sum\limits_{j_q\in\phi_m^s}{\frac{\left|h_{j_q,w}\right|^2}{d_{j_q,w}^{\alpha}}}>\frac{\mu-\sigma_w^2}{P_j}\right],
  &\mu> \sigma_w^2.
  \end{cases}
\end{align}
where (a) follows from treating the combination of the selected jammers as a random variable and applying the total probability theorem by traversing all possible combinations.

Similarly, the miss detection rate is computed as
\begin{align}\label{beta proof}
&\mathcal{P}_{MD}=\left[\sigma_c^2+\sigma_j^2+\sigma_w^2<\mu\right]\notag\\
  &
  =\begin{cases}
  0,&\mu\leq\rho_1,\\
\sum\limits_{m=0}^{N}\sum\limits_{s=1}^{\binom{N}{m}}\mathcal{P}_{\phi_m^s}\mathbb{P}\left[\sum\limits_{j_q\in\phi_m^s}\frac{\left|h_{j_q,w}\right|^2}{d_{j_q,w}^{\alpha}}<\frac{\mu-\rho_1}{P_j}\right],
  &\mu>\rho_1.
  \end{cases}
\end{align}

The proof can be completed by calculating  the inner probabilities in \eqref{alpha proof} and \eqref{beta proof} resorting to Lemma \ref{gamma_approx}.
\end{proof}

Note that $\rho_1$ in Theorem \ref{theorem_false and miss} is the sum of received signal power and noise power at W.
Theorem \ref{theorem_false and miss} shows that the noise power $\sigma_w^2$ and the signal-plus-noise power $\rho_1$ are two important boundaries for the detection at W.
Specifically, a too small detection threshold, e.g., $\mu\leq\sigma_w^2$, will lead to a complete false alarm; where if W choose $\mu$ lower than $\rho_1$, a miss detection issue can be avoid.
For other situations, the detection accuracy at W heavily depends on the aggregate jamming power level.
In particular, if no jammer is activated when T is communicating to R, a miss detection event will occur if W set $\mu>\rho_1$. That is why there exists an additional term $\mathcal{P}_{\phi_0^1}$ for $ \mathcal{P}_{MD}$ in \eqref{miss detection} compared with $\mathcal{P}_{FA}$ in \eqref{false alarm}.

After obtaining the false alarm rate $\mathcal{P}_{FA}$ and miss detection rate $\mathcal{P}_{MD}$, the detection error probability can be calculated as \eqref{detection error rate}, which is provided by the following corollary.
\begin{corollary}\label{theorem_detection_error_pro}
\textit{The detection error probability $\xi$ at W for a prescribed  detection threshold $\mu$ is given by
	\begin{align}\label{detection error total 1}
	\xi=\begin{cases}
	1,&\mu\le\sigma_w^2,\\
	\sum\limits_{m=1}^{N}\sum\limits_{s=1}^{\binom{N}{m}}\mathcal{P}_{\phi_m^s}
	\frac{\Gamma\left(v_{m,s},\frac{\mu-\sigma_w^2}{\omega_{m,s}P_j}\right)}{\Gamma\left(v_{m,s}\right)},&\sigma_w^2<\mu\le\rho_1,\\
	\mathcal{P}_{\phi_0^1}+\sum\limits_{m=1}^{N}\sum\limits_{s=1}^{\binom{N}{m}}\mathcal{P}_{\phi_m^s}
	\frac{\Gamma\left(v_{m,s},\frac{\mu-\sigma_w^2}{\omega_{m,s}P_j}\right)}{\Gamma\left(v_{m,s}\right)}
	+\\
\quad\quad	\sum\limits_{m=1}^{N}\sum\limits_{s=1}^{\binom{N}{m}}\mathcal{P}_{\phi_m^s}
\frac{\gamma\left(v_{m,s},\frac{\mu-\rho_1}{{\omega_{m,s}P}_j}\right)}{\Gamma\left(v_{m,s}\right)},&\mu>\rho_1.
	\end{cases}
	\end{align}}
\end{corollary}

Corollary \ref{theorem_detection_error_pro} shows that $\xi = 1$ when $\mu\le\sigma_w ^ 2$, which indicates a completely incorrect detection.
Some observations can be further made on how the detection performance at W is impacted by the selection threshold $\tau$ of the proposed jammer selection scheme and the jamming power $P_j$.

1) As $\tau\rightarrow\infty$, the detection error probability $\xi$ in \eqref{detection error total 1} can be simplified as
	\begin{align}\label{detection error infty}
\lim\limits_{\tau\rightarrow\infty}{\xi}=\begin{cases}
	1,&\mu\le\sigma_w^2,\\
	\frac{\Gamma\left(v_N,\frac{\mu-\sigma_w^2}{\omega_NP_j}\right)}{\Gamma\left(v_N\right)},&\sigma_w^2<\mu\le\rho_1,\\
	1-\frac{\Gamma\left(v_N,\frac{\mu-\rho_1}{\omega_NP_j}\right)-\Gamma\left(v_N,\frac{\mu-\sigma_w^2}{{\omega_NP}_j}\right)}{\Gamma\left(v_N\right)},&\mu>\rho_1.
	\end{cases}
	\end{align}
	where
	\begin{equation}\label{vN wN}
	v_N=\frac{\left(\sum_{k=1}^{N}d_{j_k,w}^{-\alpha}\right)^2}{\sum_{k=1}^{N}d_{j_k,w}^{-2\alpha}},\quad
	\omega_N=\frac{\sum_{k=1}^{N}d_{j_k,w}^{-2\alpha}}{\sum_{k=1}^{N}d_{j_k,w}^{-\alpha}}.\notag
	\end{equation}

We notice that $\tau\rightarrow\infty$ corresponds to the case where all the cooperative nodes are chosen to send jamming signals to facilitate the covert communication such that the uncertainty imposed at W by jammers vanishes.
Then, we can see from \eqref{detection error infty} that $\lim_{\tau\rightarrow\infty}{\xi}$ is solely determined on the detection threshold $\mu$ at W for a given value of $h_{t,w}$, but is independent of the selection threshold $\tau$.

2) As $P_j\rightarrow\infty$, it is easy to find that $\frac{\mu-\sigma_w^2}{\omega_{m,s}P_j}\rightarrow0$ and $\frac{\mu-\rho_1}{\omega_{m,s}P_j}\rightarrow0$.
Accordingly, we can obtain $\lim_{P_j\rightarrow\infty}\xi\rightarrow1-\mathcal{P}_{\phi_0^1}$ when $\sigma_w^2<\mu\le\rho_1$ and $\lim_{P_j\rightarrow\infty}\xi\rightarrow1$ when $\mu>\rho_1$.
This implies that the minimum detection error probability that W can achieve approaches $\xi^*=1-\mathcal{P}_{\phi_0^1}$ when a large enough jamming power is employed against surveillance.
Moreover, it is as expected that $\xi\rightarrow 1$ as $\tau,P_j\rightarrow\infty$, which can also be verified from \eqref{detection error infty}.

3) As $\tau, P_j\rightarrow 0$, which is equivalent to the case without the assistance of jammers, we can readily derive that $\xi\rightarrow 0$ when $\sigma_w^2<\mu\le\rho_1$ and $\xi\rightarrow 1$ otherwise.
This implies, if only W chooses the value of $\mu$ within $(\sigma_w^2,\rho_1]$, it can successfully detect the covert communication between T and R.

The above analyses clearly demonstrate the superiority of our proposed threshold-based jammer selection scheme.

\subsection{Minimum Average Detection Error Probability}
Since T cannot acquire the instantaneous CSI related to W $h_{t,w}$, it is impossible to calculate the exact detection error probability at W.
Instead, we focus on the average detection error probability over $h_{t,w}$.\footnote{{\color{black}As discussed in \cite{He2017On,He2018Covert}, the average detection error probability captures the average covertness performance from a Bayesian statistics perspective. There exists another approach for evaluating the covertness performance, which captures the probability that covert communication fails from an outage perspective, e.g., covert outage probability. Following \cite{He2017On,He2018Covert}, we can easily prove that the two kinds of metrics are equivalent in assessing the covertness performance of the system under investigation. }}
We consider a worst-case scenario for covert communications where the optimal detection threshold $\mu$ is designed from W's viewpoint which results in a minimum average detection error probability.
In particular, the worst-case covert communication should  consider that, for improving detection accuracy, W can adjust the detection threshold based on not only the distances between himself to T and jammers, but also the CSI of $h_{t,w}$.
In the following, we distinguish two situations depending on whether W has the knowledge of the instantaneous CSI $h_{t,w}$ or not.

\emph{Case 1: W knows the instantaneous CSI of $h_{t,w}$.}
In this situation, W is able to adjust the detection threshold $\mu$ to the change of the detection channel $h_{t,w}$ in real-time so as to achieve a minimum detection error probability.
Revisiting $\xi$ in \eqref{detection error total 1}, the minimum detection error probability under a given value of $h_{t,w}$ can be derived by $\xi^*\left(h_{t,w}\right)=\min_{\mu>0}\xi$, which turns out to be a function of a single random variable $h_{t,w}$.
Specifically, when $\sigma_w ^ 2 < \mu \le \rho_1$, we can easily prove that $\xi$ monotonically decreases from $\xi=1$ as $\mu$ increases, since $\Gamma\left(v_{m,s},\frac{\mu-\sigma_w^2}{\omega_{m,s}P_j}\right)$ decreases with $\mu$.
We can also confirm that $\xi\rightarrow1$ when $\mu\rightarrow\infty$.
This indicates that, the minimum detection error probability can be achieved at either $\mu=\rho_1$ or within the interval $\mu\in(\rho_1,\infty)$.

Based on the above discussion, the average minimum detection error probability can be calculated as
\begin{equation}\label{expected min detection}
{\bar{\xi}}_1^*=\mathbb{E}_{h_{t,w}}\left[\min_{\mu\in[\rho_1,\infty)}\xi(h_{t,w})\right].
\end{equation}

\emph{Case 2: W does not know the instantaneous CSI of $h_{t,w}$.}
In this situation, W can only design the optimal $\mu$ to minimize the average detection error probability based on the statistics of $h_{t,w}$.
Before proceeding to the calculation of the minimum average detection error probability, we first derive the expression for the average detection error probability based on \eqref{detection error total 1} by the following theorem.
\begin{theorem}\label{expected detection error_theorem}
\textit{The average detection error probability at W for a given detection threshold $\mu$ is given by
\begin{align}\label{expected detection error}
\bar{\xi}=\begin{cases}
1,&\mu\le\sigma_w^2,\\
\sum\limits_{m=1}^{N}\sum\limits_{s=1}^{\binom{N}{m}}\mathcal{P}_{\phi_m^s}\left(\frac{\Gamma\left(v_{m,s},\frac{\mu-\sigma_w^2}{\omega_{m,s}P_j}\right)}{\Gamma\left(v_{m,s}\right)}-\int_{0}^{\rho_2}e^{-x}\times\right.\\
\quad\left.\frac{\Gamma\left(v_{m,s},\frac{\mu-\rho_1(x)}{{\omega_{m,s}P}_j}\right)}{\Gamma\left(v_{m,s}\right)}dx\right)+1-e^{-\rho_2},&\mu>\sigma_w^2,
\end{cases}
\end{align}
where $\rho_1(x)\triangleq P_td_{t,w}^{-\alpha}x+\sigma_w^2$ and $\rho_2\triangleq {\left(\mu-\sigma_w^2\right)d_{t,w}^\alpha}/{P_t}$.}
\end{theorem}
\begin{proof}
We see from \eqref{miss detection} that the miss detection event only occurs when $\mu>\rho_1$, i.e., $\left|h_{t,w}\right|^2<\rho_2$.
Let $x\triangleq \left|h_{t,w}\right|^2$, then the average detection error probability for a given $\mu$ can be calculated according to \eqref{false alarm}, \eqref{miss detection}, and \eqref{detection error total 1} as
\begin{align}\label{expected detection error 2}
\bar{\xi}&=\mathbb{E}_{h_{t,w}}\left[\xi(h_{t,w})\right]=\mathcal{P}_{FA}+\mathbb{E}_{h_{t,w}}\left[\mathcal{P}_{MD}\right]\notag\\
&=\begin{cases}
1,&\mu\le\sigma_w^2,\\
\mathcal{P}_{FA}+\int_{0}^{\rho_2}{\mathcal{P}_{MD}f_{|h_{t,w}|^2}\left(x\right)dx},&\mu>\sigma_w^2,
\end{cases}
\end{align}
where $f_{|h_{t,w}|^2}\left(x\right)=e^{-x}$ is the PDF of $|h_{t,w}|^2$.
{Substituting} \eqref{false alarm} and \eqref{miss detection} into \eqref{expected detection error 2}, and the proof can be completed after some algebraic operations.
\end{proof}

{\color{black}Note that $\bar{\xi}$ in \eqref{expected detection error 2} for $\mu>\sigma_w^2$ contains two parts, where the first part reflects the negative impact of the opportunistic jammer selection on the detection accuracy at W, and the second one arises due to the absence of the instantaneous CSI of $h_{t,w}$.
	
	In the following corollary, we provide more concise expressions for $\bar{\xi}$ in \eqref{expected detection error} under the large and small selection threshold regimes, i.e., considering $\tau\rightarrow\infty$ and $\tau\rightarrow 0$, respectively.
	
	\begin{corollary}\label{corollary_error_inf}
		\textit{As $\tau\rightarrow\infty$ and $\tau\rightarrow 0$, the average detection error probability $\bar{\xi}$ for $\mu>\sigma_w^2$ approaches
			\begin{equation}\label{expected detection error infty}
			\lim\limits_{\tau\rightarrow\infty}\bar{\xi}=
			\frac{\Gamma\left(v_N,\frac{\mu-\sigma_w^2}{\omega_NP_j}\right)}{\Gamma\left(v_N\right)}+\int_{0}^{\rho_2}{e^{-x}\frac{\gamma\left(v_N,\frac{\mu-\rho_1(x)}{{\omega_NP}_j}\right)}
				{\Gamma\left(v_N\right)}dx},
			\end{equation}
			\begin{equation}
			\lim\limits_{\tau\rightarrow 0}\bar{\xi}=1-e^{-\rho_2},
			\end{equation}
			respectively, where $v_N$ and $\omega_N$ have been defined in \eqref{detection error infty}. }
	\end{corollary}
	
	Corollary \ref{corollary_error_inf} shows that as either $\tau\rightarrow\infty$ or $\tau\rightarrow 0$, $\bar{\xi}$ converges to a constant value independent of $\tau$, which implies that the uncertainty of the number of selected jammers vanishes at this time.
	
Recalling \eqref{expected detection error}, just similar to the analysis with respect to \eqref{detection error total 1}, we can prove that the optimal detection threshold $\mu$ that minimizes the average detection error probability $\bar{\xi}$ in \eqref{expected detection error} exists within the interval $\mu\in[\sigma_w^2,\infty)$.
	
	According to the above analysis, the minimum average detection error probability can be derived by
	\begin{equation}\label{expected min detection_2}
	{\bar{\xi}}_2^*=\min_{\mu\in[\sigma_w^2,\infty)}\mathbb{E}_{h_{t,w}}\left[\xi(h_{t,w})\right].
	\end{equation}
	
Combined with Corollary \ref{corollary_error_inf}, some observations can be made on the asymptotic behavior of ${\bar{\xi}}_2^*$. To be specific, we can derive that $\lim_{\tau\rightarrow\infty}{\bar{\xi}}_2^*\rightarrow 1$ as $P_j\rightarrow\infty$ and $\lim_{\tau\rightarrow\infty}{\bar{\xi}}_2^*\rightarrow 1-e^{-\rho_2}$ as $P_j\rightarrow 0$.
	Besides, we have $\lim_{\tau\rightarrow 0}{\bar{\xi}}_2^*\rightarrow 0$, which is achieved when W sets $\mu=\sigma_w^2$.
	Note that although adopting sufficiently large values of $\tau$ and $P_j$ can result in a quite poor detection performance, it will bring severe interference to R to degrade the reliability.
	These results reflect the necessary of choosing a proper value of $\tau$ in order to balance well between covertness and reliability.}

\section{Covert Throughput Maximization}
In this section, we first calculate the transmission outage probability $\delta$ of the covert communication from T to R, based on which we then design the optimal selection threshold $\tau$ and transmission rate $R$ to maximize the covert throughput subject to a covertness constraint.

\subsection{Transmission Outage Probability}
According to \eqref{yr}, the signal-to-interference-plus-noise ratio (SINR) at R is given by
\begin{equation}\label{yr SINR}
\Lambda_r=\frac{P_t\left|h_{t,r}\right|^2}{\sum_{k=1}^{N}{\mathcal{I}_{j_k}P_j\left|h_{j_k,r}\right|^2+\sigma_r^2}}.
\end{equation}

Since $\left|h_{j_k,r}\right|^2$ is a random variable to T, a transmission outage event occurs when the channel capacity  $C={\rm{log}}_2{\left(1+\Lambda_r\right)}$ falls below the transmission rate $R$, i.e.,   $C<R$.
The corresponding probability that this event happens is provided by the following lemma.
\begin{lemma}\label{P outage lemma}
\textit{The transmission outage probability from T to R is given by
\begin{equation}\label{P outage eq}
\begin{split}
\delta=&\left(1-\lambda\right)\mathcal{P}_{\phi_0^1}+\sum_{m=1}^{N}\sum_{s=1}^{\binom{N}{m}}\mathcal{P}_{\phi_m^s}\times\\
&\quad \left(1-\prod_{j_k\in\phi_m^s}\frac{\lambda\left(1-e^{-\left(\varphi P_jd_{t,r}^\alpha+d_{j_k,r}^\alpha\right)\sigma_r^2\tau}\right)}{\left(\varphi P_jd_{j_k,r}^{-\alpha}d_{t,r}^\alpha+1\right)\left(1-e^{-d_{j_k,r}^\alpha\sigma_r^2\tau}\right)}\right),
\end{split}
\end{equation}
where
$\varphi\triangleq{\left(2^R-1\right)}/{P_t}$ and
$\lambda\triangleq e^{-d_{t,r}^\alpha\varphi\sigma_r^2}$.}
\end{lemma}
\begin{proof}
Based on the definition of a transmission outage event, the transmission outage probability can be calculated as
\begin{align}\label{P outage proof}
&\delta=\mathbb{P}\left[{\rm{log}}_2{\left(1+\Lambda_r\right)}<R\right]\notag\\
&  \stackrel{\mathrm{(a)}}
=\sum_{m=0}^{N}\sum_{s=1}^{\binom{N}{m}}\mathcal{P}_{\phi_m^s}\mathbb{P}\left[\frac{\left|h_{t,r}\right|^2}{d_{t,r}^{\alpha}}<\varphi
\left(\sum_{j_k\in\phi_m^s}\frac{\left|h_{j_k,r}\right|^2}{d_{j_k,r}^{\alpha}}+\sigma_r^2\right)\right]\notag\\
&\stackrel{\mathrm{(b)}}
=\left(1-\lambda\right)\mathcal{P}_{\phi_0^1}
+\sum_{m=1}^{N}\sum_{s=1}^{\binom{N}{m}}\mathcal{P}_{\phi_m^s}\times\notag\\
&\qquad\left(1-\lambda\prod_{j_k\in\phi_m^s}{\mathbb{E}_{\left|h_{j_k,r}\right|^2}\left[e^{-\varphi P_j\left|h_{j_k,r}\right|^2d_{j_k,r}^{-\alpha}d_{t,r}^\alpha}\right]}\right)\notag\\
&=\left(1-\lambda\right)\mathcal{P}_{\phi_0^1}
+\sum_{m=1}^{N}\sum_{s=1}^{\binom{N}{m}}\mathcal{P}_{\phi_m^s}
\Bigg(1-\lambda\prod_{j_k\in\phi_m^s}\times\notag\\
&\qquad\int_{0}^{{\sigma_r^2\tau}{d_{j_k,r}^{\alpha}}}
{\check{f}}_{\left|h_{j_k,r}\right|^2}\left(y\right)e^{-\varphi P_j\left|h_{j_k,r}\right|^2d_{j_k,r}^{-\alpha}d_{t,r}^\alpha}dy\Bigg),
\end{align}
where (a) follows from the total probability theorem and (b) holds by using the CDF of ${|h_{t,r}|}^2$, i.e., $F_{\left|h_{t,r}\right|^2}\left(x\right)=1-e^{-x}$. In addition, since the selected jammers satisfy $\left|h_{j_k,r}\right|^2d_{j_k,r}^{-\alpha}/\sigma_r^2<\tau$, the PDF of $\left|h_{j_k,r}\right|^2$ can be derived as a conditional PDF given below,
\begin{align}\label{pdf h_{jk,r}}
{\check{f}}_{\left|h_{j_k,r}\right|^2}\left(y\right)=\begin{cases}
\frac{e^{-y}}{1-e^{-d_{j_k,r}^\alpha\sigma_r^2\tau}},&0\le y\le d_{j_k,r}^\alpha\sigma_r^2\tau,\\
0,&otherwise.
\end{cases}
\end{align}
Substituting \eqref{pdf h_{jk,r}} into \eqref{P outage proof} and solving the integral complete the proof.
\end{proof}

The first part in \eqref{P outage eq} is actually the transmission outage probability without jammers and the second part reflects the increased risk of causing transmission outage due to the introduction of jammers.
The following corollary further provides the transmission outage probability for the asymptotic cases with $\tau\rightarrow0$ and $\tau\rightarrow\infty$, respectively.
\begin{corollary}\label{remark out 0 inf}
\it{As the selection threshold $\tau\rightarrow0$ and $\tau\rightarrow\infty$, the transmission outage probability $\delta$ respectively approaches
\begin{equation}\label{outage 0}
\lim\limits_{\tau\rightarrow0}\delta=1-\lambda,
\end{equation}
\begin{equation}\label{outage infty}
\lim\limits_{\tau\rightarrow\infty}\delta=1-\lambda\prod_{k=1}^{N}\frac{1}{\varphi P_jd_{j_k,r}^{-\alpha}d_{t,r}^\alpha+1}.
\end{equation}}
\end{corollary}
\begin{proof}
	The results follow directly by substituting the results  that $p_k\rightarrow0$ as $\tau\rightarrow0$, and $p_k\rightarrow1$ and $e^{-\left(\varphi P_jd_{t,r}^\alpha+d_{j_k,r}^\alpha\right)\tau}\rightarrow0$ as $\tau\rightarrow\infty$. into \eqref{P outage eq}.
\end{proof}

Corollary \ref{remark out 0 inf} shows that the transmission outage probability tends to be constant as either $\tau\rightarrow0$ or $\tau\rightarrow\infty$, which will be confirmed in Fig. 5.

\subsection{Optimization of Covert Throughput}
Note that the transmission outage probability $\delta$ is a function of both the transmission rate $R$ and the selection threshold $\tau$.
The problem of maximizing the covert throughput $\Omega$ subject to a covertness constraint ${\bar{\xi}}_{\iota}^*\geq\rm{1}-\varepsilon$ for $\iota=1,2$, as shown in \eqref{R c max}, can be solved by first designing the optimal $\tau$ to minimize $\delta$ for a fixed $R$ and then deriving the optimal $R$ to achieve the globally maximum $\Omega^*$.
To begin with, the subproblem of minimizing $\delta$ can be formulated as
\begin{equation}
\min_{\tau>0}\delta(\tau), \quad
\rm{s.t.}\quad{\bar{\xi}}_{\iota}^*\geq\rm{1}-\varepsilon, \forall \iota=1,2,
\end{equation}
with the optimal $\tau^*$ characterized by the following proposition.
\begin{proposition}\label{proposition_opt_tau}
\textit{For any given transmit power $P_t$, jamming power $P_j$, total number of cooperative nodes $N$, and maximum tolerable detection probability $\varepsilon$, the optimal selection threshold $\tau^*$ for minimizing the transmission outage probability $\delta$ satisfies
\begin{equation}\label{optimal tau}
\left.{{\bar{\xi}}_\iota^*}\right|_{\tau=\tau^*}=1-\varepsilon,\quad \iota=1,2.
\end{equation}}
\end{proposition}
\begin{proof}
According to the jammer selection criteria described in \eqref{characteristic I_k}, it is obvious that using a larger value of selection threshold $\tau$ activates more jammers and causes more severe interference to R, thus resulting in a larger transmission outage probability $\delta$.
In that sense, the optimal $\tau^*$ that minimizes $\delta$ should be the minimum achievable value of $\tau$ that meets the covertness constraint ${\bar{\xi}}_{\iota}^*\geq\rm{1}-\varepsilon$ for $\iota=1,2$.
Note that the confusion on the detection at W created by jammers weakens as $\tau$ decreases, which yields a smaller minimum average detection error probability ${\bar{\xi}}_{\iota}^*$.
In other words, ${\bar{\xi}}_{\iota}^*$ is a monotonically increasing function of $\tau$.
Based on the above discussion, it is not difficult to conclude that the optimal $\tau^*$ satisfies ${\bar{\xi}}_{\iota}^*=\rm{1}-\varepsilon$, and the proof is completed.
\end{proof}

From \eqref{optimal tau}, we can conveniently obtain the optimal $\tau^*$ by calculating the inverse of ${\bar{\xi}}_{\iota}^*$ at $1-\varepsilon$, e.g., using a simple bisection method.
Afterwards, we proceed to the optimal transmission rate $R$ by solving the following subproblem.
\begin{equation}
\max_{R>0}\Omega(R)=R(1-\delta(R)),
\end{equation}
where $\delta(R)$ is an increasing function of $R$.
We can show that $\Omega\rightarrow0$ as either $R\rightarrow0$ or $R\rightarrow\infty$ due to $\delta(R)\rightarrow1$ for the latter.
This indicates that neither a too large nor a too small covert rate can produce a high covert throughput, and we should carefully design the covert rate to improve the covert throughput.
Moreover, as {will} be verified in the numerical results, the covert throughput $\Omega(R)$ is a first-increasing-then-deceasing function of $R$.
As a consequence, the optimal $R^*$ that maximizes $\Omega(R)$ can be efficiently calculated using the bisection method.

By now, we have obtained the optimal selection threshold $\tau^*$ and covert rate $R^*$, substituting which into \eqref{R c def} yields the maximum covert throughput $\Omega^*$.

\section{Numerical Results}
In this section, we present numerical results to validate our theoretical analyses and to study the achievable performance of covert communications for the multi-jammer cooperative networks.
 For simplicity, $N$ cooperative nodes are randomly generated in a circle centered at
T with a radius $2d_{t,r}$, where $d_{t,r}$ is the distance between T and R.
{\color{black}Unless specified otherwise, we assume that W does not know the instantaneous CSI of the detection channel $h_{t,w}$ between itself and T, which is rational due to the covert transmission from T.
	We define a reference distance $d_0=1000$m, and set $d_{t,r}=d_0$, $\alpha=4$, $P_t=10$dBm, $\sigma_w^2=-120$dBm, and $\sigma_r^2=-120$dBm.
	In what follows, we will first evaluate the minimum average detection error probability at W under the two scenarios where W does and does not know  $h_{t,w}$, and consider the transmission outage probability of the covert communication between T and R.
	We will then examine the influence of various key arguments on the covert throughput focusing on the case without knowing the channel $h_{t,w}$ at W.}

\begin{figure}[!t]
	\centering
	\includegraphics[width =3.5in]{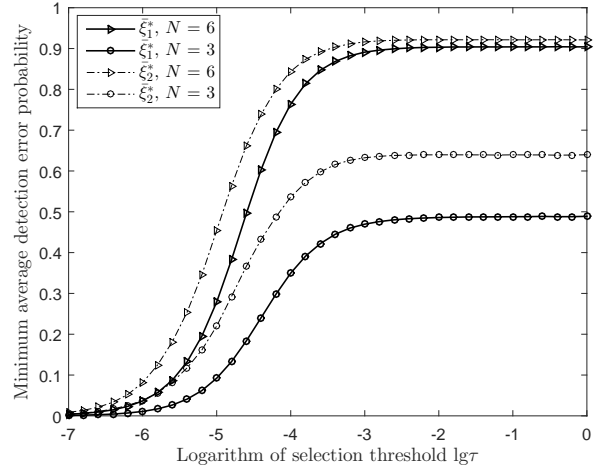}
	\caption{${\bar{\xi}}_1^*$ and ${\bar{\xi}}_2^*$ vs. lg$\tau$ for different values of $N$, with $P_j = 10$dBm and $d_{t,w}=1.2d_0$.}
	\label{ercompare tu}
	\vspace{-0.0 cm}
\end{figure}

In Fig. \ref{ercompare tu}, we aim to examine the detection performance with and without the knowledge of the instantaneous CSI of $h_{t,w}$ at W, and we illustrate the corresponding minimum average detection error probabilities ${\bar{\xi}}_1^*$ and ${\bar{\xi}}_2^*$ versus lg$\tau$ for different values of $N$.
It is {expected} that both ${\bar{\xi}}_1^*$ and ${\bar{\xi}}_2^*$ increase with $N$, since the increase of $N$ creates more randomness of the aggregate jamming power to W, which hinders its detection.
It can be seen that ${\bar{\xi}}_1^*$ is always lower than ${\bar{\xi}}_2^*$, and the gap enlarges for a smaller $N$.
This is as expected, meaning that once the instantaneous CSI of the detection channel is acquired by W, the covertness will be further compromised.

\begin{figure}[!t]
	\centering
	\includegraphics[width =3.5in]{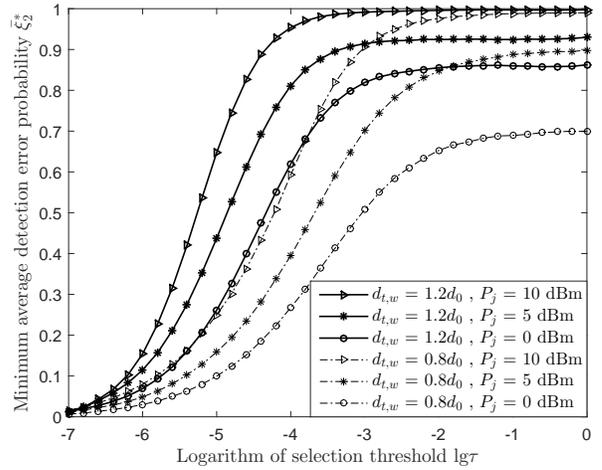}
	\caption{${\bar{\xi}}_2^*$ vs. lg$\tau$ for different values of $P_j$ and $d_{t,w}$, with $N=10$.}
	\label{er mt tu}
	\vspace{-0.0 cm}
\end{figure}

Fig. \ref{er mt tu} plots the minimum average detection error probability ${\bar{\xi}}_2^*$ versus lg$\tau$ for different values of $P_j$ and $d_{t,w}$. We note that when $\tau\rightarrow\infty$, ${\bar{\xi}}_2^*$ approaches a fixed value, which verifies Corollary \ref{corollary_error_inf}.
We find that ${\bar{\xi}}_2^*$ increases with $P_j$, since large jamming power shields the covert communication more effectively from being detected by W. We also show that ${\bar{\xi}}_2^*$ decreases significantly as $d_{t,w}$ decreases. The reason behind is that the path loss between T and W becomes smaller as $d_{t,w}$ decreases, which makes the power of covert communications much more dominate compared with the jamming power, thus remarkably improving the detection accuracy for W.

\begin{figure}[!t]
\centering
\includegraphics[width =3.5in]{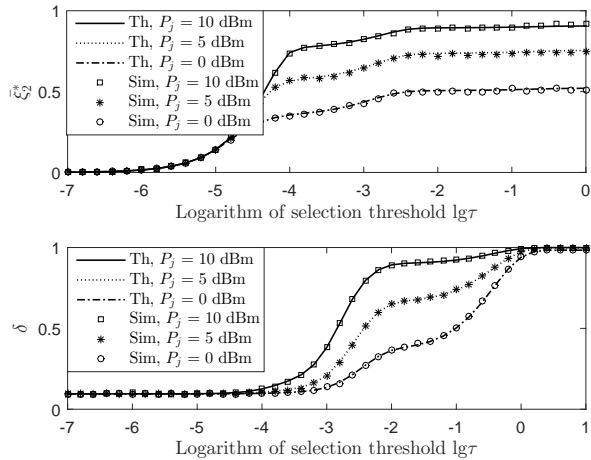}
\caption{${\bar{\xi}}_2^*$ and $\delta$ vs. lg$\tau$ for different values of $P_j$, with $N=4$ and $R=1$bits/s/Hz.}
\label{Th sim tu}
\vspace{-0.0 cm}
\end{figure}

Fig. \ref{Th sim tu} shows the minimum average detection error probability ${\bar{\xi}}_2^*$ and the transmission outage probability $\delta$ versus lg$\tau$ for different values of jamming power $P_j$. It can be seen that the Monte-Carlo simulation results match well with the theoretical values {from \eqref{expected detection error} and \eqref{P outage eq}, respectively}, which verifies the correctness of the theoretical analyses and the accuracy of the  approximation method given in \eqref{pdf Xmi}.
We can also observe that both ${\bar{\xi}}_2^*$ and $\delta$ are monotonically increasing functions of lg$\tau$, which is consistent with what is concluded in Proposition \ref{proposition_opt_tau}.

\begin{figure}[!t]
\centering
\includegraphics[width =3.5in]{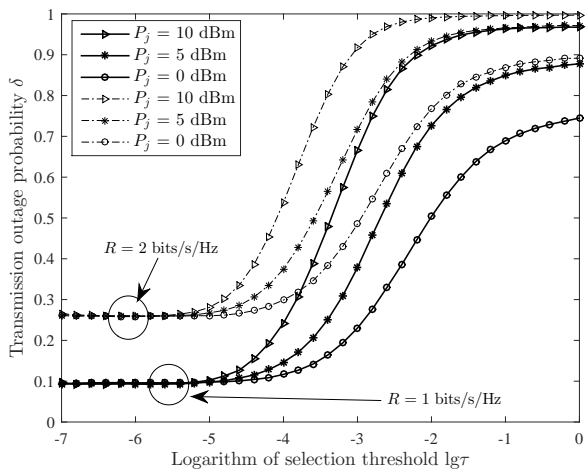}
\caption{$\delta$ vs. lg$\tau$ for different values of $P_j$ and $R$, with $N=10$ and $d_{t,w}=1.2d_0$.}
\label{out mt tu}
\vspace{-0.0 cm}
\end{figure}

Fig. \ref{out mt tu} depicts the transmission outage probability $\delta$ versus lg$\tau$ for different jamming power $P_j$ and covert rate $R$.
We observe that $\delta$ monotonically increases with  lg$\tau$, $R$, and $P_j$, and approaches a constant value as either $\tau\rightarrow0$ or $\tau\rightarrow\infty$, which is consistent with Corollary \ref{remark out 0 inf}.
In addition, the convergence value of $\delta$ as $\tau\rightarrow0$ is solely determined by $R$ and is independent of $P_j$, since no jammer is likely to be activated for a small enough $\tau$.

\begin{figure}[!t]
\centering
\includegraphics[width =3.5in]{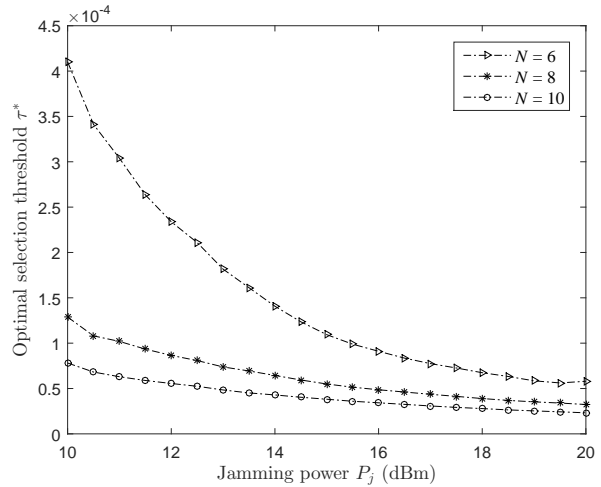}
\caption{$\tau^*$ vs. $P_j$ for different values of $N$, with $\varepsilon=0.1$.}
\label{tm mt tu}
\vspace{-0.0 cm}
\end{figure}

Fig. \ref{tm mt tu} shows the optimal selection threshold $\tau^*$ versus $P_j$ for different values of $N$.
We show that $\tau^*$ is a monotonically decreasing function of $N$ and $P_j$.
This is because the uncertainty of differentiating desired signal and jamming power increases with $N$ and $P_j$, which produces a larger detection error probability at W.
As a result, smaller $\tau$ is allowed to satisfy a certain covertness constraint for larger $N$ and $P_j$.

\begin{figure}[!t]
\centering
\includegraphics[width =3.5in]{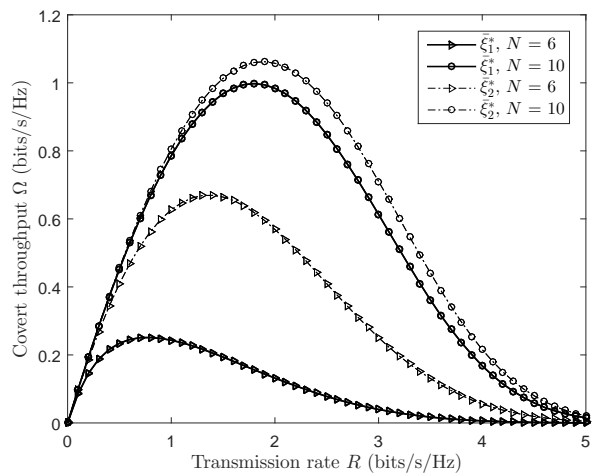}
\caption{$\Omega$ vs. $R$ for different values of $N$, with $P_j=10$dBm, $d_{t,w}=1.2d_0$, and $\varepsilon=0.1$.}
\label{Rccompare tu}
\vspace{-0.0 cm}
\end{figure}

Fig. \ref{Rccompare tu} illustrates the covert throughput $\Omega$ versus the covert rate $R$ for different numbers of cooperative nodes $N$, considering the situations with and without the instantaneous CSI of the detection channel $h_{t,w}$.
We see that $\Omega$ increases first and then decreases as $R$ increases, and approaches zero as either $R\rightarrow0$ or $R\rightarrow\infty$.
It can also be observed that $\Omega$ increases as $N$ becomes larger.
This is because the increase of $N$ will undoubtedly create more randomness of total jamming power and degrade the detection accuracy at W, whereas the transmission outage probability between T and R will not increase significantly owing to the jammer selection criterion.
 We also find that $\Omega$ with ${\bar{\xi}}_2^*$ is larger than that with ${\bar{\xi}}_1^*$ particularly for the small $N$ regime, since a higher level of covertness can be promised when W is unaware of the instantaneous CSI of the channel from T to itself.

\begin{figure}[!t]
	\centering
	\includegraphics[width =3.5in]{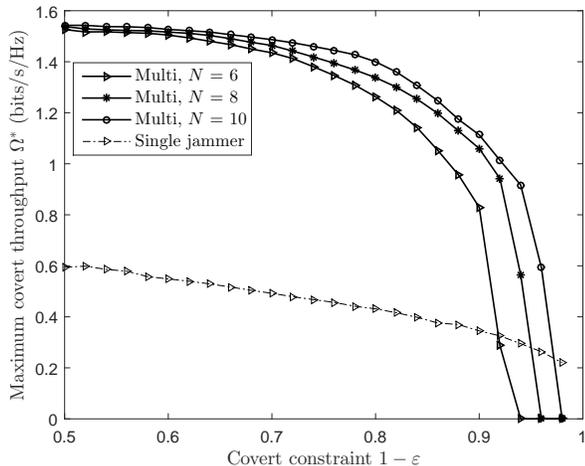}
	\caption{$\Omega^*$ vs. $1-\varepsilon$ for different values of $N$, with $P_j=10$dBm and $d_{t,w}=1.2d_0$.}
	\label{covertness tu}
	\vspace{-0.0 cm}
\end{figure}
Fig. \ref{covertness tu} plots the maximum covert throughput $\Omega^*$ versus $1-\varepsilon$ for different numbers of cooperative nodes $N$.
We also adopt the single-jammer scheme adopted in \cite{Sobers2017Covert} as the benchmark scheme, where variable jamming power is exploited to confuse the detector.
Comparing the two schemes, we show that $\Omega^*$ for the multi-jammer scheme significantly outperforms  that for the single-jammer scheme for a quite wide range of $\varepsilon$.
This demonstrates the superiority of our proposed multi-jammer scheme in terms of covert communications.
Nevertheless, we can observe that the single-jammer scheme appears to be superior for an extremely small value of $\varepsilon$ which corresponds to a high level of covertness requirement.
The underlying reason is that the randomness of total jamming power created by a certain number of cooperative nodes might not be sufficient to satisfy a stringent enough covert constraint, whereas the single-jammer scheme can always guarantee uncertainty at W by enlarging the jamming power budget and adopting random power continuously.

\begin{figure}[!t]
\centering
\includegraphics[width =3.5in]{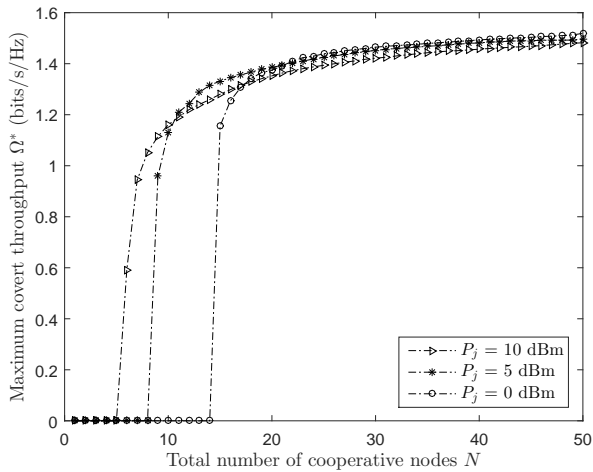}
\caption{$\Omega^*$ vs. $N$ for different values of $P_j$, with $d_{t,w}=1.2d_0$ and $\varepsilon=0.1$.}
\label{Rcm mt tu}
\vspace{-0.0 cm}
\end{figure}

Fig. \ref{Rcm mt tu} depicts the maximum covert throughput $\Omega^*$ versus $N$ for different values of $P_j$.
{It is as expected that $\Omega^*$ increases with $N$, and we can see that as $N$ continuously increases, the improvement of covert throughput becomes less remarkable.
This suggests that a moderate number of cooperative nodes should be favorable for assisting the covert communication.}
In addition, there exists a minimum number $N_{min}$ of cooperative nodes below which the covertness constraint cannot be satisfied under a certain level of jamming power, and the value of $N_{min}$ becomes smaller as $P_j$ increases, e.g., $N_{min}= 5$ for $P_j=10$dBm and $N_{min}= 8$ for $P_j=5$dBm.
Interestingly, we also observe that a smaller (larger) $P_j$ can produce a higher $\Omega^*$ when $N$ is large (small) enough, which is attributed to the trade-off between covertness and reliability.

\section{Conclusion}
This paper investigated covert communications assisted by multiple friendly cooperative nodes over slow fading channels, in which transmitter T aims to covertly deliver messages to receiver R under the surveillance of warden W, while a portion of the cooperative nodes will be intelligently selected as jammers to send interference signals to confuse W.
We proposed an opportunistic jammer selection scheme, based on which we analyzed the average detection error probability at W and the transmission outage probability of the covert communication between T and R.
We then jointly designed the optimal selection threshold and the covert rate to maximize the covert throughput subject to a certain covertness constraint.
By comparing with the single-jammer case, numerical results demonstrated the superiority of our proposed multi-jammer scheme in terms of covert throughput.
We showed that there exists an optimal jammer selection threshold for maximizing the covert throughput, and the optimal value decreases as either the number of cooperative nodes or the jamming power increases.
We also revealed that a minimum number of cooperative nodes should be deployed in order to meet a certain level of covertness requirement; as the number of cooperative nodes continuously increases, the covert throughput first significantly improves but then reaches saturation.

\appendix
{\color{black}\subsection{Proof of the Sufficiency of Radiometer}\label{proof_sufficiency_radiometer}
To justify the use of a radiometer as the detector, we prove that the test statistic $T_w$ in \eqref{T_e} is sufficient for W's detection.

The received signal at W in \eqref{y_e} can be rewritten as
\begin{equation}\label{received signal at W}
\bm{y}_w[i]=\sqrt{P_t}\frac{h_{t,w}}{d_{t,w}^{\alpha/2}}\bm{x}_t[i]\theta+
\sum\limits_{k=1}^{N}{\mathcal{I}_{j_k}\sqrt{P_j}\frac{h_{j_k,w}}{d_{j_k,w}^{\alpha/2}}\bm{v}_{j_k}[i]}+\bm{n}_w[i],
\end{equation}
where $\theta$ represents the indicator of whether T is transmitting or not, with $\theta = 1$ meaning that T is transmitting while $\theta = 0$ otherwise.
Since we assume that $\bm{x}_t [i],\bm{v}_{j_k} [i]\sim\mathcal{CN} (0, 1)$, $\bm{n}_w [i]\sim\mathcal {CN} (0, \sigma_w ^2)$, we have
$\bm{y}_w [i]\sim\mathcal {CN} (0, \sigma_\theta ^2)$, where $\sigma_\theta ^2=\sigma_c^2\theta+\sigma_j^2+\sigma_w^2$ with $\sigma_c^2$ and $\sigma_j^2$ defined in \eqref{T_e}.
Since W is unaware of the instantaneous CSI between ${\rm J}_k$ and R, the number of selected jammers is not available at W such that $\sigma_j^2$ becomes a random variable to W.
However, the PDF $f_{\sigma_{j}^{2}}(x)$ of $\sigma_j^2$ is assumed known at W.

According to the Fisher-Neyman factorization theorem, W's test statistic $T_w=(1/n)\sum_{i=1}^{n}\left|\bm{y}_w[i]\right|^2$ is sufficient for the indicator $\theta$ if and only if the conditioned probability $\mathcal{P}(\bm{y}_w|\theta)$ depends on $\bm{y}_w$ only through $T_w$.

From \eqref{received signal at W}, $\mathcal{P}(\bm{y}_w|\theta)$ can be calculated as
\begin{align}
\mathcal{P}(\bm{y}_w|\theta) &=\int_0^{\infty} \prod_{i=1}^{n} \left\{\frac{1}{\sqrt{2 \pi \sigma_{\theta}^{2}}}e^{-\frac{\left|\bm{y}_{w}[i]\right|^{2}}{2 \sigma_{\theta}^{2}}}\right\}  f_{\sigma_{j}^{2}}\left(x\right) dx \notag\\
&=\int_0^{\infty}\left(2 \pi \sigma_{\theta}^{2}\right)^{-\frac{n}{2}} e^{-\frac{n}{2 \sigma_{\theta}^{2}} T_{w}} f_{\sigma_{j}^{2}}\left(x\right) dx.
\end{align}
We clearly show that $\mathcal{P}(\bm{y}_w|\theta)$ relies on $\bm{y}_w$ only via $T_w$ regardless of the form of  $f_{\sigma_{j}^{2}}\left(x\right)$. This completes the proof.
}

\subsection{Proof of Lemma \ref{gamma_approx}}\label{proof_lemma_2}
The $y$-th cumulant of a random variable $X_{m,s}$ is defined as
\begin{equation}\label{Q Xmi i}
Q_{X_{m,s}}^{\left(y\right)}=\left.\frac{d^y\mathbb{E}_{X_{m,s}}\left[e^{\omega X_{m,s}}\right]}{d\omega^y}\right|_{\omega=0}.
\end{equation}

The mean and variance of $X_{m,s}$ are $\mu_{X_{m,s}}=Q_{X_{m,s}}^{\left(1\right)}$ and $\sigma_{X_{m,s}}^2=Q_{X_{m,s}}^{\left(2\right)}-\left(Q_{X_{m,s}}^{\left(1\right)}\right)^2$, respectively.
Hence, the parameters $v_{m,s}$ and $\omega_{m,s}$ can be calculated as $v_{m,s}={\mu_{X_{m,s}}^2}/{\sigma_{X_{m,s}}^2}$ and $\omega_{m,s}={\sigma_{X_{m,s}}^2}/{\mu_{X_{m,s}}}$. Recall that $X_{m,s}=\sum_{j_k\in\phi_m^s}{\left|h_{j_k,w}\right|^2d_{j_k,w}^{-\alpha}}$. Due to the mutual independence among $\left\{h_{j_k,w}\right\}$ for $j_k\in\phi_m^s$, we can calculate $\mathbb{E}_{X_{m,s}}\left[e^{\omega X_{m,s}}\right]$ for a deterministic $\phi_m^s$ as follows,
\begin{equation}\label{E Xmi}
\mathbb{E}_{X_{m,s}}\left[e^{\omega X_{m,s}}\right]=\prod_{j_k\in\phi_m^s}{\mathbb{E}_{\left|h_{j_k,w}\right|^2}\left[e^{-\omega\left|h_{j_k,w}\right|^2d_{j_k,w}^{-\alpha}}\right]}.
\end{equation}

Denote $g_{j_k,w}\left(\omega\right)=\mathbb{E}_{\left|h_{j_k,w}\right|^2}\left[e^{-\omega\left|h_{j_k,w}\right|^2d_{j_k,w}^{-\alpha}}\right]$ such that $\mathbb{E}_{X_{m,s}}\left[e^{\omega X_{m,s}}\right]=\prod_{j_k\in\phi_m^s}{g_{j_k,w}\left(\omega\right)}$ and $Q_{X_{m,s}}^{\left(y\right)}=\left.\left({d^y\prod_{j_k\in\phi_m^s}{g_{j_k,w}\left(\omega\right)}}/{d\omega^y}\right)\right|_{\omega=0}$.
Then, we have $\left.g_{j_k,w}\left(\omega\right)\right|_{\omega=0}=1$, $\left.\left({dg_{j_k,w}\left(\omega\right)}/{d\omega}\right)\right|_{\omega=0}=d_{j_k,w}^{-\alpha}$, and $\left.\left({d^2g_{j_k,w}\left(\omega\right)}/{d\omega^2}\right)\right|_{\omega=0}=2d_{j_k,w}^{-2\alpha}$. Based on these results, the first and second cumulants of $X_{m,s}$ can be respectively calculated as
\begin{equation}\label{Q Xmi 1}
\begin{split}
Q_{X_{m,s}}^{\left(1\right)}
&=
\sum_{j_k\in\phi_m^s}\prod_{j_p\in\phi_m^s\backslash j_k}\left.\left(g_{j_p,w}\left(\omega\right)\frac{dg_{j_k,w}\left(\omega\right)}{d\omega}\right)\right|_{\omega=0}\\
&=\sum_{j_k\in\phi_m^s} d_{j_k,w}^{-\alpha},
\end{split}
\end{equation}
\begin{align}\label{Q Xmi 2}
Q_{X_{m,s }}^{\left(2\right)}&=\sum_{j_k\in\phi_m^s}\left(\prod_{j_p\in\phi_m^s\backslash j_k}g_{j_p,w}\left(\omega\right)\frac{d^2g_{j_k,w}\left(\omega\right)}{d\omega^2}
+\frac{dg_{j_k,w}\left(\omega\right)}{d\omega}\right.\notag\\
&\quad\left.\left.\times\sum_{j_p\in\phi_m^s\backslash j_k}\prod_{j_q\in\phi_m^s\backslash \left\{j_k,j_p\right\}}g_{j_q,w}
\left(\omega\right)\frac{dg_{j_p,w}\left(\omega\right)}{d\omega}\right)\right|_{\omega=0}\notag\\
&=\sum_{j_k\in\phi_m^s}\left(2d_{j_k,w}^{-2\alpha}+d_{j_k,w}^{-\alpha}\sum_{j_p\in\phi_m^s\backslash j_k} d_{j_p,w}^{-\alpha}\right).
\end{align}

After some algebraic manipulations, we can obtain $\mu_{X_{m,s}}=\sum_{j_k\in\phi_m^s} d_{j_k,w}^{-\alpha}$ and $\sigma_{X_{m,s}}^2=\sum_{j_k\in\phi_m^s} d_{j_k,w}^{-2\alpha}$, substituting which into $v_{m,s}$ and $\omega_{m,s}$ completes the proof.

\bibliographystyle{IEEEtran}
\bibliography{covertbib}

\end{document}